\documentclass{article}
\usepackage{graphicx}
\usepackage{epstopdf}
\usepackage{amsmath, amsthm, amssymb}
\usepackage{amsfonts}
\usepackage{graphicx}
\usepackage{subcaption} 
\usepackage[margin=1in]{geometry}
\usepackage{color}

\newtheorem{theorem}{Theorem}

\newtheorem{prop}[theorem]{Proposition}
\newtheorem{corollary}[theorem]{Corollary}
\theoremstyle{definition}

\theoremstyle{remark}

\usepackage[english]{babel}

\usepackage{amsmath}
\usepackage{graphicx}
\usepackage[colorlinks=true, allcolors=blue]{hyperref}

\renewcommand{\geq}{\geqslant}
\renewcommand{\leq}{\leqslant}

\title{An $n$th-cousin mating model and the $n$-anacci numbers}

\author{Elisa Heinrich Mora\thanks{Department of Biology, Stanford University, Stanford, CA 94305 USA.}\hspace{0.2cm} and Noah A.~Rosenberg$^*$}

\begin{document}
\maketitle

\begin{abstract}
\noindent In seeking to understand the size of inbred pedigrees, J.~Lachance (\emph{J.~Theor.~Biol.~261, 238--247, 2009)} studied a population model in which, for a fixed value of $n$, each mating occurs between $n$th cousins. We explain a connection between the second-cousin case of the model ($n=2$) and the Fibonacci sequence, and more generally, between the $n$th-cousin case and the $n$-anacci sequence $(n \geq 2)$. For a model with $n$th-cousin mating $(n \geq 1)$, we obtain the generating function describing the size of the pedigree $t$ generations back from the present, and we use it to evaluate the asymptotic growth of the pedigree size. In particular, we show that the growth of the pedigree asymptotically follows the growth rate of the $n$-anacci sequence --- the golden ratio $\phi = (1 + \sqrt{5})/2 \approx 1.6180$ in the second-cousin case $n=2$ --- and approaches 2 as $n$ increases. The computations explain the appearance of familiar numerical sequences and constants in a pedigree model. They also recall similar appearances of such sequences and constants in studies of population biology more generally.
\end{abstract}


\section{Introduction}
\label{sec:intro}

In organisms in which individuals possess two parents, the biparental genealogy of an individual --- the pedigree --- can be treated as a mathematical structure. Pedigrees as mathematical structures provide the basis for a variety of applications in demography, genealogical studies, and genetics. For example, mathematical properties of pedigrees contribute to understanding the level of relationship of individuals in a population, the genealogical history of populations, the time to the most recent common ancestor of a set of individuals, and aspects of individual genotypes and phenotypes~\cite{AgranatTamirEtAl24:genetics, Chang99, EdgeAndCoop20, ErlichEtAl18, RohdeEtAl04, Thompson00, WakeleyEtAl12, WakeleyEtAl16}.

In one example of a mathematical study of pedigrees, Lachance \cite{lachance2009inbreeding} examined a discrete-time model in which, for a fixed value of $n$, each mating in a population occurs between $n$th cousins. Lachance considered an individual member of the population, counting the number of distinct genealogical ancestors that the individual possesses $t$ generations back in time. In generations $t=1, 2, \ldots, n+1$, the number of ancestors is $2^t$. However, because of the $n$th-cousin matings, starting in generation $t=n+2$, some ancestors are ancestors via multiple paths in the pedigree, so that \emph{distinct} genealogical ancestors number fewer than $2^t$.   

Lachance obtained a recurrence describing the number of distinct genealogical ancestors in generation $t$ in the $n$th-cousin model, evaluating the first terms of sequences in $t$ counting the numbers of ancestors for small values of $n$. For small $n$, Lachance also numerically computed the growth rates with $t$ of the numbers of genealogical ancestors, informally observing a numerical equivalence between these rates and the Fibonacci, tribonacci, tetranacci, and pentanacci growth constants for $n=2$, 3, 4, and 5, respectively.

Here, we provide mathematical explanations for Lachance’s observations. We show that the values produced by Lachance’s recurrence for the number of distinct genealogical ancestors in the $n$th-cousin model are equivalent to those obtained in a modification of the $n$-nacci recurrence. We derive the generating function for the number of distinct genealogical ancestors in the $n$th-cousin model, showing that it is closely related to the $n$-nacci generating function, sharing the same exponential growth constant $r_{n+1}$. Finally, we show that as $n$ increases, the constant $r_{n+1}$ that describes the asymptotic generation-by-generation growth of the number of distinct genealogical ancestors in generation $t$ increases monotonically from $\phi=(1+\sqrt{5})/2$ in the second-cousin model ($n=2$) toward 2 in the limit as $n \rightarrow \infty$. We discuss connections to classical results of population modeling that also produce similar appearances of the Fibonacci sequence, and of $n$-anacci sequences more generally.

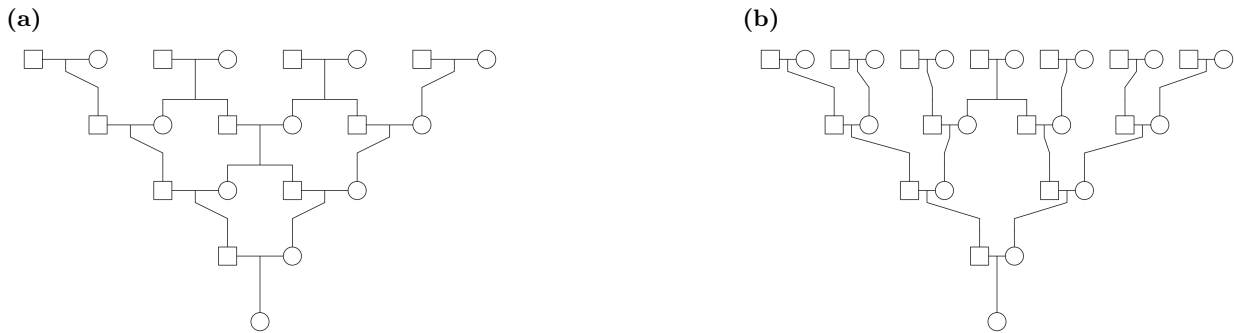
\begin{figure}[htbp]
    \centering
    \begin{subfigure}[t]{0.4\textwidth}
        \begin{picture}(0,0)
            \put(0,120){\small \textbf{(a)}}
        \end{picture}
        \includegraphics[width=\textwidth]{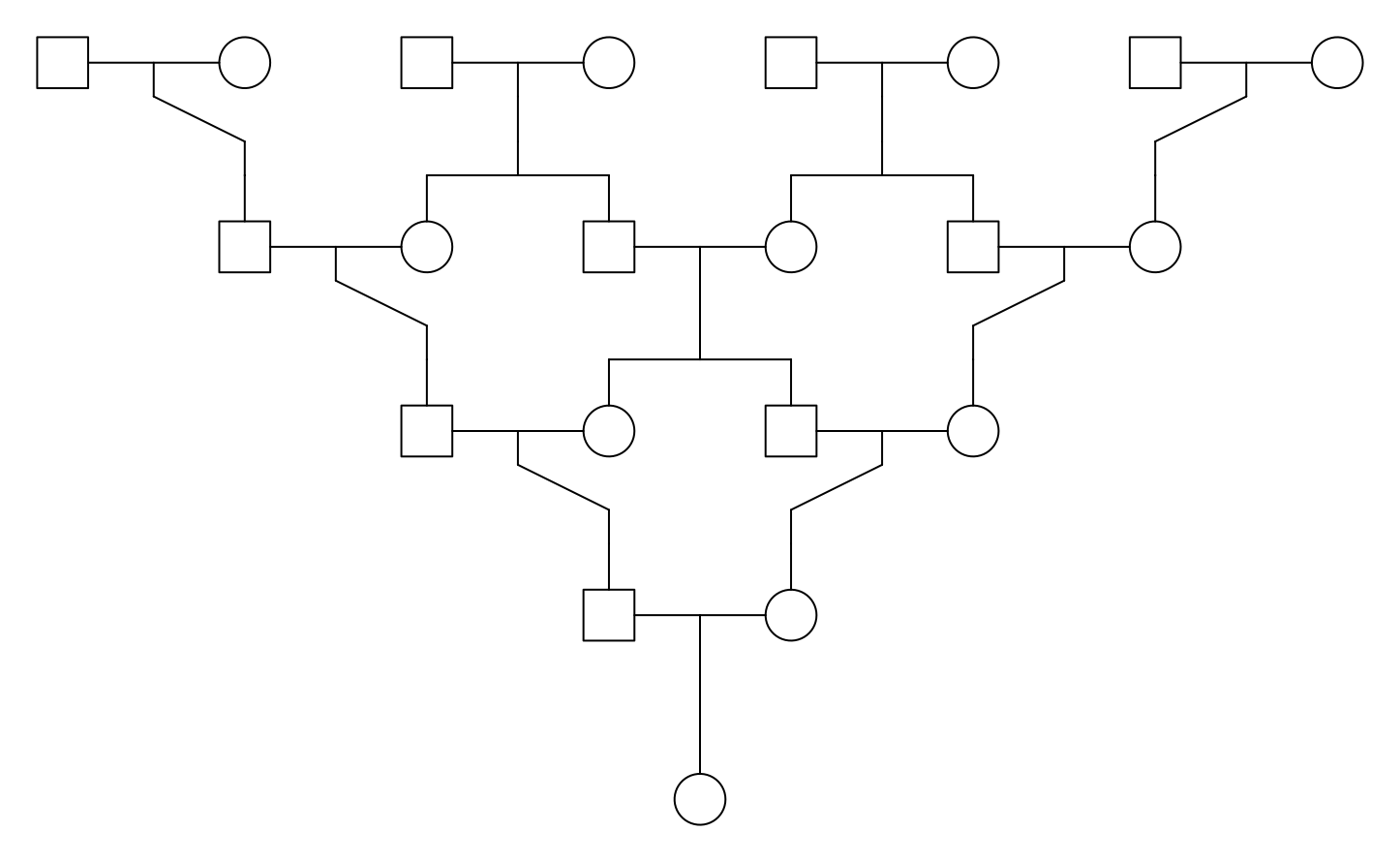}
    \end{subfigure}
    \hfill
    \begin{subfigure}[t]{0.4\textwidth}
        \begin{picture}(0,0)
            \put(0,120){\small \textbf{(b)}}
        \end{picture}
        \includegraphics[width=\textwidth]{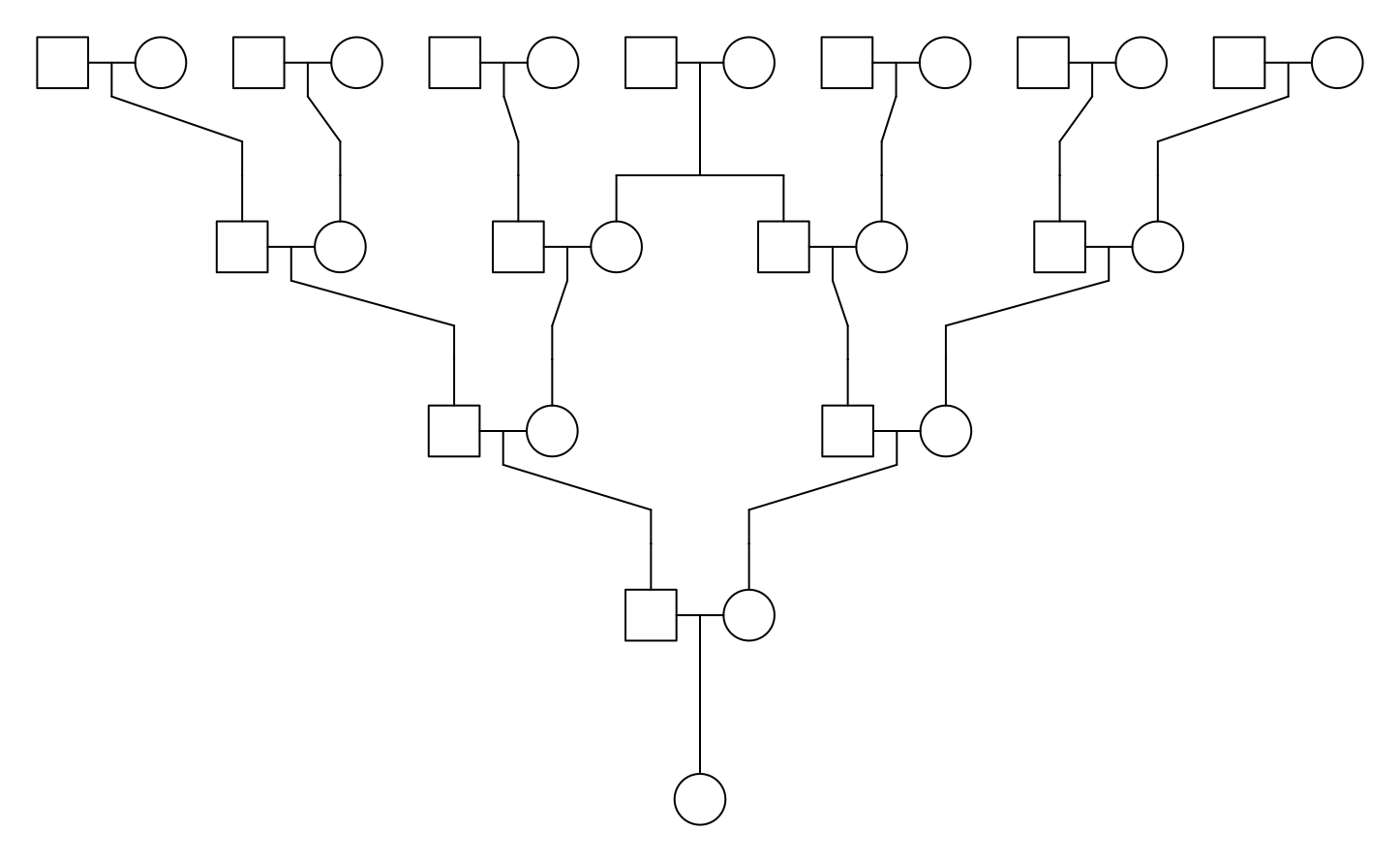}
    \end{subfigure}
    \vspace{.3cm}
    \caption{Pedigrees with $n$th-cousin mating, in which the members of each mating pair are $n$th cousins, possessing a pair of shared ancestors $g=n+1$ generations back in time. (a) First-cousin mating ($n=1$, $g=2$). (b) Second-cousin mating ($n=2$, $g=3$). Males are represented by squares, and females by circles.}
    \label{fig:cousin_inbreeding}
\end{figure}

\section{Model and preliminary results}

\subsection{Model}
\label{sec:model}

We use the discrete-generation $n$th-cousin mating model of Lachance~\cite{lachance2009inbreeding}. The model considers a large biparental population of equal sex ratio. Mating is nonrandom, consisting of cousin matings of a fixed type. In particular, each mating event involves $n$th-cousin mating for some specified value of $n \geq 1$, with $n=0$ corresponding to sib mating.

It is convenient to write $g=n+1$, where $g$ denotes the number of generations separating an $n$th-cousin pair from their most recent shared pair of ancestors; for example, a set of grandparents ($g=2$) for first cousins ($n=1$). We let $t$ denote the number of generations back in a pedigree from an index individual; $t=0$ is the index individual, $t=1$ is the parental generation, $t=2$ is the grandparental generation, and so on. We denote by $A_t$ the number of distinct ancestors represented in generation $t \geq 0$.

Figure \ref{fig:cousin_inbreeding} depicts pedigrees with $n$th-cousin mating for $n=1$ and $n=2$. Figure \ref{fig:pedigree_panels} illustrates the model in more detail in the case of first cousins: $n=1$, corresponding to $g=2$. We observe that the number of distinct ancestors in generation $t=2$ is $2^2=4$. In generation $t=3$, because of the first-cousin mating in generation $t=1$, the number of distinct ancestors is 6, strictly less than $2^3=8$.

\subsection{Preliminary results}
\label{pedigree_const} 

In the model in which each individual is the offspring of an $n$th-cousin mating pair $(n \geq 0)$, this section recapitulates the recurrence from Lachance~\cite{lachance2009inbreeding} that counts genealogical ancestors in the ancestral generation $t$ for an individual from generation 0. First, in an outbred population without $n$th-cousin mating, the number of ancestors doubles each generation, so that for $t \geq 1$,
\begin{equation}
A_t = 2A_{t-1}.
\end{equation}
With $A_0=1$, this recurrence has solution $A_t=2^t$.

In an inbred population with $n$th-cousin mating $(n \geq 0)$, the number of distinct ancestors is $2^t$ in the most recent generations but eventually decreases as certain ancestors begin to appear in the pedigree along multiple genealogical lines. By definition of $n$th-cousin mating, each mating pair shares a common ancestral pair $g=n+1$ generations prior. Starting at generation $t=0$ and moving backward, the first mating pair appears in generation $t=1$, so that the first \emph{duplicated} pair occurs in generation $t=g+1$. Subtracting a duplicated pair in the $g$th generation preceding a mating pair --- or one duplicated ancestor for each individual in the mating pair --- the recurrence takes the form~\cite[eq.~1]{lachance2009inbreeding}:
\begin{equation}
A_t = 
\begin{cases} 
    2^t, & \text{if } 0 \leq t \leq g, \\
    2A_{t-1} - A_{t-g}, & \text{if } t \geq g+1.
\end{cases}
\label{eq:recursion}
\end{equation}

\begin{figure}[ht]
     \centering
    \begin{subfigure}{0.7 \textwidth} 
        \centering
        \includegraphics[width=\textwidth]{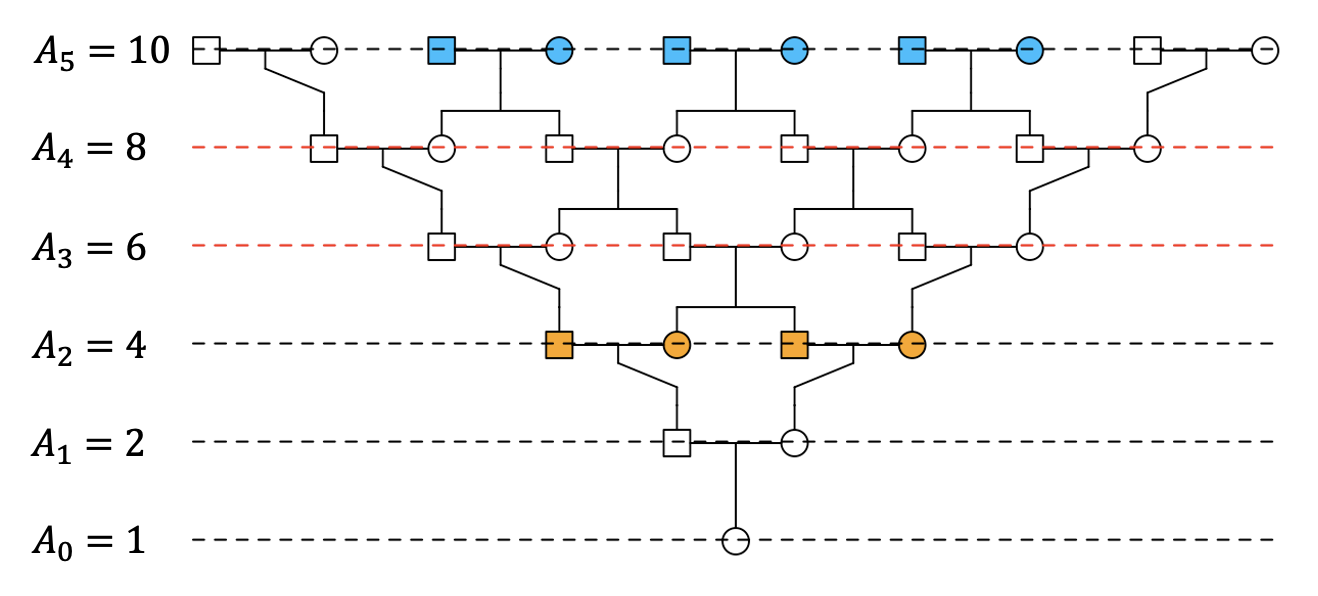}
        \label{fig:panel_c}
    \end{subfigure}
    \vspace{-.6cm}
    \caption{Illustration of notation for a pedigree with first-cousin mating ($n=1$, $g=2$). For a generation $t$, here $t=5$, inbred pairs in generation $t-3=2$ are indicated in yellow, shared ancestors of the members of the inbred pairs from generation $t-3$ are indicated in blue in generation $t$, and the intervening generations are indicated in red. The first value of $t$ back in time at which the number of distinct ancestors is less than $2^t$ is $t=n+2=g+1=3$, with $A_3=6$.}
    \label{fig:pedigree_panels}
\end{figure}

Lachance iterated eq.~\ref{eq:recursion} to obtain the number of ancestors $A_t$ for the smallest values of $t$, for $n$th-cousin mating with small values of $n$ (Table~\ref{tab:table1}). Lachance~\cite[eq.~3]{lachance2009inbreeding} also noted that for $n \geq 1$, the asymptotic growth rate of the number of ancestors, $r_g=\lim_{t \rightarrow \infty} (A_t/A_{t-1})$, can be obtained by finding the largest real root of the characteristic polynomial of eq.~\ref{eq:recursion}, 
\begin{equation}
\label{eq:characteristic}
    r^g = 2r^{g-1}-1.
\end{equation}
Lachance \cite[section 4.1]{lachance2009inbreeding} observed that for $n \geq 2$, the pedigree growth rates $r_g$ under different degrees of cousin mating, as seen in Table \ref{tab:table1}, matched those of corresponding $n$-anacci sequences ($g=n+1$). The growth rate $r_3$ is the growth rate of the Fibonacci sequence, $r_4$ is the growth rate of the tribonacci sequence, $r_5$ is the growth rate of the tetranacci sequence, and so on. In emphasizing other aspects of the pedigree model, Lachance~\cite{lachance2009inbreeding} did not further discuss the connections to $n$-anacci sequences; we supply explanations for those connections here.

\section{A Fibonacci-like recurrence} 
\label{rec_bt}

The Fibonacci sequence $\{f_t\}_{t \geq 0}$ is defined $f_0=0$, $f_1=1$, and $f_t = f_{t-1}+f_{t-2}$ for $t \geq 2$. The generalized Fibonacci sequence $\{f_t^{(n)}\}_{t \geq 0}$ of order $n$, $n \geq 2$, or the $n$-anacci sequence (or sometimes $n$-bonacci), has $f_0^{(n)}=0$, $f_1^{(n)}=1$, $f_t^{(n)}=2^{t-2}$ for $2 \leq t \leq n-1$, and $f_t^{(n)}=f_{t-1}^{(n)} + f_{t-2}^{(n)} + \ldots + f_{t-n}^{(n)}$ for $t \geq n$~\cite[p.~101]{CullEtAl05}. The Fibonacci sequence is thus the $n$-anacci sequence of order 2. With some lack of standardization in the indexing of the terms, the 3-, 4- and 5-anacci sequences --- the tribonacci, tetranacci, and pentanacci sequences, respectively --- appear in the On-Line Encyclopedia of Integer Sequences as sequences A000073, A000078, and A001591.

We now define a new sequence $\{B_t\}_{t \geq 0}$ with a recurrence reminiscent of Fibonacci and $n$-anacci recurrences. This sequence is defined purely in mathematical terms, without reference to the mating model. However, it includes a positive integer parameter $g=n+1$ separating two regimes. We will show an equivalence between $B_t$ and $A_t$, facilitating analysis of the properties of $A_t$.

For an integer $n \geq 0$ with $g = n+1$, define $\{B_t\}_{t \geq 0}$ as follows:
\begin{equation}
B_t = 
\begin{cases} 
2^t, & \text{if } 0 \leq t \leq g, \\
\left(\sum\limits_{i=1}^{n} B_{t-i}\right) + 2, & \text{if } t \geq g+1.
\end{cases}
\label{eq:Bt_recursion}
\end{equation}
The value of $n$ represents the number of terms summed in $B_t$ in the recursive regime, which begins at step $g+1=n+2$. The case of $n=0$ is trivial, with $B_0=1$ and $B_t=2$ for all $t \geq 1$. 

\begin{prop}
\label{prop:AequalsB}
Let \( g \geq 2 \). Define the sequences \( \{A_t\}_{t \geq 0} \) and \( \{B_t\}_{t \geq 0} \) as follows:
\begin{align*}
A_t &= 
\begin{cases}
2^t, & \text{if } 0 \leq t \leq g, \\
2A_{t-1} - A_{t-g}, & \text{if } t \geq g+1,
\end{cases} \\
B_t &= 
\begin{cases}
2^t, & \text{if } 0 \leq t \leq g, \\
\bigg( \sum \limits_{i=1}^{g-1} B_{t-i} \bigg)+ 2, & \text{if } t \geq g+1.
\end{cases}
\end{align*}
Then for all \( t \geq 0 \), we have \( A_t = B_t \).
\end{prop}

\begin{table}[tb]
\centering
\begin{tabular}{|c|c|c|c|c|c|c|c|}
\hline
      &          & \multicolumn{6}{|c|}{Cousin mating level ($n$)} \\ \hline
$t$   & Outbred ($n=\infty$) & $n=0$ & $n=1$ & $n=2$ & $n=3$ & $n=4$ & $n=5$ \\ \hline
1     & 2        &  2 & 2 &    2 &   2 &   2 &   2 \\ 
2     & 4        &  2 & 4 &    4 &   4 &   4 &   4 \\ 
3     & 8        &  2 & 6 &    8 &   8 &   8 &   8 \\ 
4     & 16       &  2 & 8 &   14 &  16 &  16 &  16 \\ 
5     & 32       &  2 & 10 &  24 &  30 &  32 &  32 \\
6     & 64       &  2 & 12 &  40 &  56 &  62 &  64 \\
7     & 128      &  2 & 14 &  66 & 104 & 120 & 126 \\
8     & 256      &  2 & 16 & 108 & 192 & 232 & 248 \\
9     & 512      &  2 & 18 & 176 & 354 & 448 & 488 \\
10    & 1024     &  2 & 20 & 286 & 652 & 864 & 960 \\ \hline
$g$   & -        &  1 & 2 &   3 &   4 &   5 &   6 \\ \hline 
$r_g$ & 2        &  1 & 1 & 1.6180 & 1.8393 & 1.9276 & 1.9659 \\ \hline
\end{tabular}
\caption{Number of distinct genealogical ancestors as a function of the number of generations $t$ back in time from an index individual, under different levels of $n$th-cousin mating. With $g=n+1$, values of $r_g$, the growth rate of the number of ancestors with time, are obtained numerically by solving eq.~\ref{eq:characteristic}. The table is adapted from Table 1 of \cite{lachance2009inbreeding}.}
\label{tab:table1}
\end{table}

\begin{proof}
The claim is trivial for $0 \leq t \leq g$, as both sequences are explicitly defined $A_t = B_t = 2^t$ (eqs.~\ref{eq:recursion} and \ref{eq:Bt_recursion}). For $t \geq g+1$, we use induction. For the base case, $t=g+1$, we have 
\begin{align}
    A_{g+1} & = 2A_g-A_1 = 2(2^g -1), \nonumber \\
    B_{g+1} & =\bigg( \sum_{i=1}^{g-1}
    B_{g+1-i} \bigg) + 2 = \bigg(\sum_{i=1}^{g-1} 2^{g+1-i} \bigg) + 2 = \sum_{i=1}^g 2^g = 2(2^g-1). \nonumber
    \end{align}
For the inductive step, suppose $A_t = B_t$ for all $t$, $0 \leq t \leq g+x$, for some $x \geq 1$. We show $A_{g+x+1}=B_{g+x+1}$.

From the definition of $A_t$ (eq.~\ref{eq:recursion}), we have:
$$A_{g+x+1} = 2A_{g+x} - A_{x+1}.$$
From the inductive hypothesis and the definition of $B_t$ (eq.~\ref{eq:Bt_recursion}), we have
\begin{align}
A_{g+x+1} &= 2B_{g+x} - B_{x+1} \nonumber \\
&= 2\bigg[ \bigg( \sum_{i=1}^{g-1} B_{g+x-i} \bigg) + 2 \bigg] - B_{x+1} \nonumber \\
&= 2\bigg[ \bigg( \sum_{i=1}^{g-1} A_{g+x-i} \bigg) + 2 \bigg]- A_{x+1} 
\nonumber \\
&= 2 \bigg[ \bigg( \sum_{i=1}^{g-1} A_{x+i} \bigg) + 2 \bigg] - A_{x+1}. \nonumber
\end{align}
From the definition of $B_t$ (eq.~\ref{eq:Bt_recursion}) and the inductive hypothesis, 
\begin{align}
B_{g+x+1} &= \bigg( \sum_{i=1}^{g-1} B_{g+x+1-i} \bigg) + 2 \nonumber  \\
&= \bigg( \sum_{i=1}^{g-1} A_{g+x+1-i} \bigg) + 2 \nonumber \\
&= \bigg( \sum_{i=2}^g A_{x+i} \bigg) + 2. \nonumber 
\end{align}
Simplifying and using the inductive hypothesis, we have
\begin{align}
A_{g+x+1}-B_{g+x+1} &= \bigg( 2\sum_{i=1}^{g-1} A_{x+i} \bigg) + 4 - A_{x+1}  - \bigg( \sum_{i=2}^g A_{x+i} \bigg) - 2 \nonumber \\
&= \bigg( \sum_{i=2}^{g-1} A_{x+i} \bigg) + 2A_{x+1} + 2 -A_{x+1} - A_{x+g}\nonumber \\
&= \nonumber \bigg[ \bigg( \sum_{i=1}^{g-1} A_{x+i} \bigg) + 2 \bigg]  - A_{x+g} \\
&= \bigg[ \bigg( \sum_{i=1}^{g-1} B_{x+i} \bigg) + 2 \bigg]  - B_{x+g} \nonumber \\
&= \bigg[ \bigg( \sum_{i=1}^{g-1} B_{x+g-i} \bigg) + 2 \bigg]  - B_{x+g}.
\label{eq:laststep}
\end{align} 
Eq.~\ref{eq:laststep} is equal to 0 by the definition of $B_t$ (eq.~\ref{eq:Bt_recursion}), completing the proof.
\end{proof}

We immediately obtain a corollary that in the second-cousin model ($n=2$, $g=3$), the coefficients $A_t$ are closely related to the Fibonacci sequence. 
\begin{corollary}
The sequence $\{A_t\}_{t \geq 1}$ in the case of $g=3$ satisfies $A_t = 2f_{t+2}-2$, where $\{f_t\}_{t \geq 0}$ is the Fibonacci sequence $f_0=0$, $f_1=1$, and $f_t=f_{t-1}+f_{t-2}$ for $t \geq 2$.  
\end{corollary}
\begin{proof}
The result holds for $t=1$, 2, and 3, with $(f_3, f_4, f_5) = (2, 3, 5)$ and $A_1=2f_3-2=2$, $A_2=2f_4-2=4$, and $A_3=2f_5-2=8$.

The base case for induction is $t=4$, for which $f_6=8$ and $A_4=2f_6-2=14$. For the inductive step, suppose $A_{3+x}=2f_{5+x}-2$ for all positive $x\leq c$. Then $A_{3+(x+1)}=A_{2+(x+1)}+A_{1+(x+1)}+2$ by the equivalence of $A_t$ and $B_t$ and the definition of $B_t$ (Proposition \ref{prop:AequalsB}), so that by the inductive hypothesis, $A_{3+(x+1)}=[2f_{4+(x+1)} -2] + [2f_{3+(x+1)} - 2] + 2 = 2f_{5+(x+1)} - 2$, completing the proof.
\end{proof}

\section{Generating function for \texorpdfstring{$A_t$}{At}}

We use the equivalence between $A_t$ and $B_t$ to analyze the growth rate $r_g$ of the number of ancestors in the $n$th-cousin mating model, $g=n+1$, with $n \geq 1$. To do so, we derive the generating function for $\{A_t\}_{t \geq 0}$.

\begin{prop}
Write $G_g(z) = \sum_{t=0}^\infty A_t z^t$ for the generating function for the sequence $\{A_t\}_{t \geq 0}$ with parameter $g$. Then
\begin{equation*}
G_g(z) = \frac{1 + z^g}{1 - 2z + z^g}.
\end{equation*}
\end{prop}

\begin{proof}
Recalling from eq.~\ref{eq:recursion} that $A_t = 2^t$ for $0 \leq t \leq g$, we have $G_g(z) = \sum_{t=0}^g 2^t z^t + \sum_{t=g+1}^\infty A_t z^g$. The first term is $\sum_{t=0}^g 2^t z^t = [1-(2z)^{g+1}]/(1-2z)$.

For the second term, following the recurrence in eq.~\ref{eq:recursion} for $t \geq g+1$, we have
\begin{align}
\sum_{t=g+1}^\infty A_t z^t &= 2 \sum_{t=g+1}^\infty A_{t-1} z^t - \sum_{t=g+1}^\infty A_{t-g} z^t. \nonumber \\
&= 2z \sum_{t=g}^\infty A_t z^t - z^g \sum_{t=1}^\infty A_t z^t. \nonumber
\end{align}
We substitute $\sum_{t=g}^\infty A_t z^t = G_g(z) - [1-(2z)^{g+1}]/(1-2z) + 2^g z^g$ and $\sum_{t=1}^\infty A_t z^t = G_g(z)-1$, obtaining
\begin{align*}
\sum_{t=g+1}^\infty A_t z^t &= 2z \bigg[G_g(z) - \frac{1-(2z)^{g+1}}{1-2z} + 2^g z^g\bigg] - z^g \big[ G(z) - 1 \big]. \nonumber
\end{align*}
Adding $\sum_{t=0}^g 2^t z^t$, we have 
\begin{align*}
G_g(z) = \sum_{t=0}^\infty A_t z^t &= \frac{1-(2z)^{g+1}}{1-2z} + 2z \bigg[G_g(z) - \frac{1-(2z)^{g+1}}{1-2z} + 2^g z^g\bigg] - z^g \big[ G_g(z) - 1 \big].
\end{align*}
Solving for $G_g(z)$ and simplifying yields the result.
\end{proof}

The generating function $G_g(z)$ can be written 
\begin{equation}
\label{eq:gf}
    G_g(z) = \frac{1+z^g}{z(1-z)} F_n(z),
\end{equation}
where $n=g-1$ and 
\begin{equation}
\label{eq:fn}
    F_n(z) = \frac{z}{1 - z - z^2 - \ldots - z^n}.
\end{equation}
$F_n(z)$ is the generating function for an $n$-anacci sequence~\cite[p.~270]{Knuth98}. For example, $F_2(z) = z/(1-z-z^2)$ is the generating function for the classic Fibonacci (2-anacci) sequence $f_n$~\cite[p.~297]{GrahamEtAl94}. The generating function $G_3(z)$ is the product of $(1+z^3)/[z(1-z)]$ and the Fibonacci generating function. That $G_g(z)$ can be written as a product of an $n$-anacci generating function and a second term $(1+z^g)/[z(1-z)]$ explains why the growth rate $r_g$ follows the growth rate of the $n$-anacci sequence. 

\section{Asymptotic analysis of pedigree growth}
For $n \geq 1$, to find the asymptotic growth of the coefficients $A_t$ in the $n$th-cousin mating model, $r_g = \lim_{t \rightarrow \infty} (A_{t}/A_{t-1})$, we find the singularity nearest the origin~\cite[section 5.2]{Wilf94} of generating function $G_g(z)$ in eq.~\ref{eq:gf} (and hence of $F_n(z)$ in eq.~\ref{eq:fn}). This singularity occurs at $z=1/r_g$. For first cousins ($g=2$), $r_2=1$, and for second cousins ($g=3$), $r_3$ is the familiar $(1 + \sqrt{5})/2$. 

Table \ref{tab:table1} reports numerical solutions for $r_g$ for $g=4, 5, 6$. We can observe informally that as $n$ increases, $r_g=r_{n+1}$ increases as well. We establish formally that the singularity of $F_n(z)$ of smallest modulus decreases toward $\frac{1}{2}$ as $n$ increases, so that the generation-wise growth rates $r_g$ increase toward 2.
\begin{prop}
\label{prop:singularity}
   For $n \geq 1$, denote by $c_n$ the singularity of $F_n(z)$ of smallest modulus. Then (i) $c_{n+1} < c_n$, and (ii) $\lim_{n \rightarrow \infty} c_n = \frac{1}{2}$. 
\end{prop}
\begin{proof}
(i) First, $c_1=1$. For $n \geq 2$, let $Q_n(z)$ denote the denominator of $F_n(z)$, $Q_n(z)=1-z-z^2-\ldots-z^n$. Then $Q_{n+1}(c_n)=1 - c_n - c_n^2 - \ldots - c_n^n - c_n^{n+1} = -c_n^{n+1}$, as $c_n$ is a root of $Q_n(z)$. We have $Q_{n+1}(\frac{1}{2}) > 0$ and $Q_{n+1}(1)<0$, so that by the intermediate value theorem, a root for $Q_{n+1}$ must exist in $(\frac{1}{2},1)$. $Q_{n+1}(z)$ decreases monotonically in $(\frac{1}{2},1)$. Because $Q_{n+1}(c_n)<0$, a root for $Q_{n+1}$ must occur in $(\frac{1}{2}, c_n)$. 

(ii) To show $\lim_{n \rightarrow \infty} c_n = \frac{1}{2}$, choose $\epsilon$ with $0 < \epsilon < \frac{1}{2}$. We show that there exists $n$ such that $c_n < \frac{1}{2} + \epsilon$. First, $Q_n(\frac{1}{2}) = 2^{-n} > 0$. Next, $Q_n(\frac{1}{2} + \epsilon) = 1 - (\frac{1}{2} + \epsilon) - (\frac{1}{2} + \epsilon)^2 - \ldots - (\frac{1}{2} + \epsilon)^n = 2^{-n} - P_n(\epsilon)$, where $P_n(\epsilon)$ is a polynomial in $\epsilon$ with positive coefficients and no constant term. In particular, $P_n(\epsilon) > \epsilon$ for the permissible $\epsilon$ values, $0 < \epsilon < \frac{1}{2}$.

Then $Q_n(\frac{1}{2} + \epsilon) = 2^{-n} - P_n(\epsilon) < 2^{-n} - \epsilon$. Choose $n > -\log_2 \epsilon$, so that $Q_n(\frac{1}{2} + \epsilon) < 2^{-n} - \epsilon < 0$. Because $Q_n(\frac{1}{2}) > 0$ and $Q_n(\frac{1}{2}+\epsilon) < 0$, using the intermediate value theorem, we have shown that given $\epsilon$, for sufficiently large $n$, the root $c_n$ lies in $(\frac{1}{2}, \frac{1}{2}+\epsilon)$. We then take $\epsilon$ to 0. 
\end{proof}

\section{Discussion}

Previously, Lachance~\cite{lachance2009inbreeding} had made a number of informal observations about the growth of pedigree sizes in a model in which each individual is the offspring of an $n$th-cousin mating pair, for a fixed $n \geq 1$. We have formalized several of these findings. First, Lachance~\cite{lachance2009inbreeding} noticed that for $n \geq 2$ and $g=n+1$, growth rates $r_g$ corresponded to the growth rates of associated $n$-anacci sequences. To explain this observation, we have shown an equivalence between the recurrence $A_t$ counting the number of ancestors in a pedigree under the $n$th-cousin mating model and a recurrence $B_t$ that resembles the $n$-anacci recurrence (Proposition~\ref{prop:AequalsB}) and has its same growth rate. We have also proven that the growth rate $r_g$ approaches 2 as $n \rightarrow \infty$ (Proposition \ref{prop:singularity}). In obtaining these results, we have reported a generating function for the number of ancestors $A_t$ (eq.~\ref{eq:gf}). 

Our mathematical analysis is similar to previous analyses of $n$-anacci sequences. An equation similar to eq.~\ref{eq:recursion} for tabulating $n$-anacci sequences in a recurrence containing only two of the terms appears in~\cite[p.~2]{NoeAndVosPost05}. The limit of 2 for the constants $r_g$ (Proposition \ref{prop:singularity}) is studied by~\cite{Flores67}. The setting here contains specific transformations of $n$-anacci sequences; the transformations of the Fibonacci and tribonacci sequences in Table \ref{tab:table1} appear in the On-Line Encyclopedia of Integer Sequences as sequences A019274 and A054668, respectively. 

Beyond formalizing Lachance's observations~\cite{lachance2009inbreeding}, the study adds to the mathematical connections between kinship systems and Fibonacci-like integer sequences. The Fibonacci sequence famously traces to a problem of population biology~\cite{Bacaer11}: the evaluation of the growth of a rabbit population under specified mating rules. It arises in studies of the number of ancestors that contribute to X chromosomes~\cite{Basin63, BuffaloEtAl16}. It also arises in certain problems in the analysis of heterozygosity in animal and plant breeding models that trace to the earliest years of the study of genetics~\cite{Crow87, Rosenberg16}. That the second-cousin mating model produces Fibonacci counts for the number of ancestors in pedigrees adds to the scenarios in which the Fibonacci sequence appears in counting problems involving pedigrees and population growth. 

Less well known is that a \emph{tribonacci} sequence like that of the third-cousin model (Table \ref{tab:table1}) has arisen in evaluating population size in a problem of historic interest, namely Darwin's problem of the population growth of elephants. In a passage of the \emph{Origin of Species}, Darwin describes a hypothetical model for elephant population growth; later correspondence suggests that Darwin's assumptions give rise to a tribonacci model for the growth~\cite{Podani2018}. To our knowledge, the third-cousin ancestor count is a rare example of a biological problem for which a sequence related to the tribonacci sequence provides the solution. 

We have focused here on specific forms of $n$th-cousin mating. Other systems of potential interest could include matings between members of different generations, and an example provided by Lachance~\cite{lachance2009inbreeding} illustrates such a scenario, with Fibonacci growth. We expect that inter-generational systems of mating can be analyzed in a manner similar to that considered here.


\vskip .5cm \noindent {\bf Acknowledgments.} We acknowledge grant support from National Science Foundation grant BCS-2116322.

{\small
\bibliographystyle{plain}
\bibliography{references}

\begin{thebibliography}{10}

\bibitem{AgranatTamirEtAl24:genetics}
L.~Agranat-Tamir, J.~A. Mooney, and N.~A. Rosenberg.
\newblock Counting the genetic ancestors from source populations in members of an admixed population.
\newblock {\em Genetics}, 226:iyae011, 2024.

\bibitem{Bacaer11}
N.~{Baca{\"e}r}.
\newblock {\em A Short History of Mathematical Population Dynamics}.
\newblock Springer, London, 2011.

\bibitem{Basin63}
S.~Basin.
\newblock The {Fibonacci} sequence as it appears in nature.
\newblock {\em Fibonacci Quarterly}, 1:53--56, 1963.

\bibitem{BuffaloEtAl16}
V.~Buffalo, S.~M. Mount, and G.~Coop.
\newblock A genealogical look at shared ancestry on the {X} chromosome.
\newblock {\em Genetics}, 204:57--75, 2016.

\bibitem{Chang99}
J.~T. Chang.
\newblock Recent common ancestors of all present-day individuals.
\newblock {\em Advances in Applied Probability}, 31:1002--1026, 1999.

\bibitem{Crow87}
J.~F. Crow.
\newblock Seventy years ago in {{\it Genetics}}: {H.~S.~Jennings} and inbreeding theory.
\newblock {\em Genetics}, 115:389--391, 1987.

\bibitem{CullEtAl05}
P.~Cull, M.~Flahive, and R.~Robson.
\newblock {\em Difference Equations: From Rabbits to Chaos}.
\newblock Springer, New York, 2005.

\bibitem{EdgeAndCoop20}
M.~D. Edge and G.~Coop.
\newblock Donnelly (1983) and the limits of genetic genealogy.
\newblock {\em Theor. Pop. Biol.}, 133:23--24, 2020.

\bibitem{ErlichEtAl18}
Y.~Erlich, T.~Shor, I.~Pe'er, and S.~Carmi.
\newblock Identity inference of genomic data using long-range familial searches.
\newblock {\em Science}, 362:690--694, 2018.

\bibitem{Flores67}
I.~Flores.
\newblock Direct calculation of $k$-generalized {F}ibonacci numbers.
\newblock {\em Fibonacci Quarterly}, 5:259--266, 1967.

\bibitem{GrahamEtAl94}
R.~L. Graham, D.~E. Knuth, and O.~Patashnik.
\newblock {\em Concrete Mathematics: A Foundation for Computer Science}.
\newblock Addison-Wesley, Boston, 2nd edition, 1994.

\bibitem{Knuth98}
D.~E. Knuth.
\newblock {\em The Art of Computer Programming, Volume 3 Sorting and Searching}.
\newblock Addison-Wesley, Upper Saddle River, NJ, 2nd edition, 1998.

\bibitem{lachance2009inbreeding}
J.~Lachance.
\newblock Inbreeding, pedigree size, and the most recent common ancestor of humanity.
\newblock {\em Journal of Theoretical Biology}, 261:238--247, 2009.

\bibitem{NoeAndVosPost05}
T.~D. Noe and J.~{Vos Post}.
\newblock Primes in {F}ibonacci $n$-step and {L}ucas $n$-step sequences.
\newblock {\em Journal of Integer Sequences}, 8:05.4.4, 2005.

\bibitem{Podani2018}
J.~Podani, {\'A}.~Kun, and A.~Szil{\'a}gyi.
\newblock How fast does {D}arwin’s elephant population grow?
\newblock {\em Journal of the History of Biology}, 51:259--281, 2018.

\bibitem{RohdeEtAl04}
D.~L.~T. Rohde, S.~Olson, and J.~T. Chang.
\newblock Modelling the recent common ancestry of all living humans.
\newblock {\em Nature}, 431:562--566, 2004.

\bibitem{Rosenberg16}
N.~A. Rosenberg.
\newblock Admixture models and the breeding systems of {H.~S.~Jennings}.
\newblock {\em Genetics}, 202:9--13, 2016.

\bibitem{Thompson00}
E.~A. Thompson.
\newblock {\em Statistical Inferences from Genetic Data on Pedigrees}.
\newblock Institute for Mathematical Statistics, Beachwood, OH, 2000.

\bibitem{WakeleyEtAl12}
J.~Wakeley, L.~King, B.~S. Low, and S.~Ramachandran.
\newblock Gene genealogies within a fixed pedigree, and the robustness of {Kingman's} coalescent.
\newblock {\em Genetics}, 190:1433--1445, 2012.

\bibitem{WakeleyEtAl16}
J.~Wakeley, L.~King, and P.~R. Wilton.
\newblock Effects of the population pedigree on genetic signatures of historical demographic events.
\newblock {\em Proc. Natl. Acad. Sci. USA}, 113:7994--8001, 2016.

\bibitem{Wilf94}
H.~S. Wilf.
\newblock {\em Generatingfunctionology}.
\newblock Academic Press, San Diego, 2nd edition, 1994.

\end{thebibliography}
}\end{document}